\documentclass[12pt,reqno]{amsart}
\pdfoutput=1 
\usepackage{amsmath,bbm}
\usepackage{latexsym}
\usepackage{amsfonts}
\usepackage{amssymb}
\usepackage{color}
\usepackage[table]{xcolor}

\usepackage{graphicx}
\usepackage{url}
\usepackage{enumerate}
\usepackage{tikz}
\usetikzlibrary{shadings, intersections, calc, plotmarks}
\usepackage{marginnote}
\usepackage{stackrel}
\usepackage{empheq}
\usepackage{tcolorbox}
\usetikzlibrary{arrows.meta,positioning}

\usepackage{geometry}
\usepackage{hyperref}
\hypersetup{
  colorlinks=true,
  linkcolor=blue!90!black,    
  citecolor=green!50!black,   
  urlcolor=blue!70!black,     
  linktoc=all,                
}
\usepackage{cleveref}

\xdefinecolor{tumblue}     {RGB}{0,101,189}
\xdefinecolor{tumgreen}    {RGB}{162,173,  0}
\xdefinecolor{tumred}      {RGB}{229, 52, 24}
\xdefinecolor{tumivory}    {RGB}{218,215,203}
\xdefinecolor{tumorange}   {RGB}{227,114, 34}
\xdefinecolor{tumlightblue}{RGB}{152,198,234}

\newtheorem*{proposition*}{Proposition}

\newtheorem*{theorem*}{Theorem}

\newtheorem*{lemma*}{Lemma}

\newtheorem*{corollary*}{Corollary}

\newcommand{\R}{\mathbbm{R}}

\newcommand{\C}{\mathbbm{C}}

\newcommand{\N}{\mathbbm{N}}

\newcommand{\1}{\mathbbm{1}}

\def\>{{\rangle}}
\def\<{{\langle}}

\newcommand{\be}{\begin{equation}}
\newcommand{\ee}{\end{equation}}
\newcommand{\bea}{\begin{eqnarray}}
\newcommand{\eea}{\end{eqnarray}}

\newcommand{\tr}[1]{\operatorname{tr}\mathopen{}\left[#1\right]\mathclose{}} 
\newcommand{\Tr}{\mathrm{tr}}

\newcommand{\overbar}[1]{\mkern3mu\overline{\mkern-3mu#1\mkern-1.5mu}\mkern1.5mu} 

\begin{document}

\title{Ky Fan majorization for binary tensor products}

\author[Wolf]{Michael M. Wolf$^{1,2}$}
\email{m.wolf@tum.de}
\address{$^1$ Department of Mathematics, Technical University of Munich}
\address{$^2$ Munich Center for Quantum
Science and Technology (MCQST),  Munich, Germany}
\author[Zhou]{Yuanheng Zhou$^{1}$}


\begin{abstract} We provide a short proof of a Ky Fan-type majorization relation for the singular values of a sum of binary tensor products of matrices. This generalizes Alhejji's result \cite{Alhejji2025} from two summands to arbitrary sums and from positive matrices to arbitrary matrices. As an application, we show a majorization relation between the singular values of a completely positive map and those of its Kraus operators.
\end{abstract}
\maketitle

 Ky Fan's majorization relation for Hermitian matrices \cite[Ch.\,9]{MOmajo} states that the eigenvalues of a sum of matrices are majorized by the sum of the ordered eigenvalues. Recently, Alhejji  \cite{Alhejji2025} extended this, for positive semidefinite matrices, to a sum of two binary tensor products:
\begin{equation}\label{eq:intro}
    \begin{aligned}
    \lambda\Big(\sum_{l=1}^m A_l\otimes B_l\Big)&\prec\sum_{l=1}^m \lambda(A_l)\otimes\lambda(B_l),\quad \text{for }m=2.\\
\end{aligned}
\end{equation}
In this note, we show that Alhejji's result for binary tensor products holds more generally: for arbitrary sums ($m\in\N$), and arbitrary  matrices. \vspace*{4pt}

We begin by fixing some notation.
For a matrix $A\in\C^{d\times d}$, we will write $\sigma_1(A)\geq\cdots\geq\sigma_d(A)$ for the nonincreasingly ordered list of singular values, and if $A$ is Hermitian, $\lambda_1(A)\geq\cdots\geq\lambda_d(A)$ for the eigenvalues.  We will write $\1$ for the identity matrix, whose dimension will depend on the context.

Two vectors $x,y\in\R^d$ satisfy the \emph{weak majorization} relation $x\prec_w y$ if \smash{$\sum_{i=1}^k x_i^\downarrow\leq \sum_{i=1}^k y_i^\downarrow$} holds for all $k\in\{1,\ldots,d\}$. Here, the down arrows denote a nonincreasingly ordered arrangement of the entries. \emph{Majorization} $x\prec y$ requires that, in addition, equality holds for $k=d$.\vspace*{3pt}

Our proof consists of three steps: (i) we employ the partial trace $\Tr_1$ to eliminate one tensor factor, (ii) on the remaining tensor factor, we follow the classical proof strategy \cite{KyFan} to show the result for positive semidefinite matrices, (iii) we extend it to singular values of arbitrary matrices by a Cauchy-Schwarz argument already used in \cite{RicoWolf}. 
\begin{lemma*}\label{lemma} Let $P\in\C^{d_1\times d_1}$ be a Hermitian rank-$r$ projector, $B\in\C^{d_2\times d_2}$ positive semidefinite, and $E\in\C^{d_1\times d_1}\otimes \C^{d_2\times d_2}$ satisfying $0\leq E\leq \1$. Then
    \begin{equation}
        \tr{E(P\otimes B)}\leq\sum_{j=1}^{d_2}\min\{r,\lambda_j(R)\}\lambda_j(B),\quad\text{where } R\coloneq\Tr_1[E].
    \end{equation}
\end{lemma*}
\begin{proof}
      $R_P:=\Tr_1[E(P\otimes\1)]=\Tr_1[(\sqrt{P}\otimes\1)E(\sqrt{P}\otimes\1)]$ is a partial trace of a positive semidefinite operator and thus $R_P\geq 0$. Moreover, it is bounded by $R_P\leq r\1$ since $E\leq \1$ and $\tr{P}=r$, but also by $R_P\leq R$ since $R=R_P+R_{\1-P}$. So Weyl's monotonicity theorem \cite[Cor.\,III.2.3]{Bhatia} gives $\lambda_j(R_P)\leq\min\{r,\lambda_j(R)\}$. Finally, we can use the definition of $R_P$ together with von Neumann's trace inequality \cite[7.4.1]{horn13} to obtain
      $$\tr{E(P\otimes B)}=\tr{R_P B}\leq\sum_{j=1}^{d_2}\lambda_j(R_P)\lambda_j(B)\leq\sum_{j=1}^{d_2}\min\{r,\lambda_j(R)\}\lambda_j(B).\vspace*{-7pt}$$
\end{proof}
Now we can state and prove our main result:
\begin{theorem*}\label{theorem}
    Let $m\in\mathbbm{N}$, $A_l\in\C^{d_1\times d_1}$,  $B_l\in\C^{d_2\times d_2}$ for $l\in\{1,\ldots,m\}$. Then
    \begin{equation}\label{eq:main}
        \sigma\Big(\sum_{l=1}^m A_l\otimes B_l\Big)\ \prec_w\ \sum_{l=1}^m\sigma(A_l)\otimes\sigma(B_l).
    \end{equation}
    If all matrices are positive semidefinite, then singular values become eigenvalues and weak majorization becomes majorization.
\end{theorem*}
\begin{proof}
    We will first prove Eq.\,\eqref{eq:main} for positive semidefinite matrices. For an arbitrary $A\geq 0$, let \smash{$A=\sum_{i=1}^{d_1} \lambda_i(A)|\psi_i\rangle\langle\psi_i|$} be a spectral decomposition with orthonormal $\psi_i$. With $P_r:=\sum_{i=1}^r |\psi_i\rangle\langle\psi_i|$, we can rewrite this as
 $$A=\sum_{r=1}^{d_1}\underbrace{\big(\lambda_r(A)-\lambda_{r+1}(A)\big)}_{\geq 0} P_r,\qquad \text{setting }\lambda_{d_1+1}(A)\coloneqq 0.$$
    For an arbitrary $B\geq 0$ and $0\leq E\leq\1$ we can now use the Lemma and obtain\vspace*{-5pt}
    \begin{equation*}
    \begin{aligned}
        \tr{E(A\otimes B)}&=\sum_{r=1}^{d_1} \big(\lambda_r(A)-\lambda_{r+1}(A)\big) \tr{E(P_r\otimes B)}\\
        &\leq\sum_{r=1}^{d_1}\big(\lambda_r(A)-\lambda_{r+1}(A)\big)\sum_{j=1}^{d_2}\min\{r,\lambda_j(R)\}\lambda_j(B)\\
        &=\sum_{r=1}^{d_1}\sum_{j=1}^{d_2}\lambda_r(A)\lambda_j(B)\underbrace{\big(\min\{r,\lambda_j(R)\}-\min\{r-1,\lambda_j(R)\}\big)}_{\eqqcolon C_{r,j}}
    \end{aligned}
    \end{equation*}
    The matrix $C$ only depends on $E$,   satisfies $C_{r,j}\in[0,1]$, and $\sum_{r,j}C_{r,j}=\tr{E}$. The latter follows from \smash{$\sum_{r=1}^{d_1} C_{r,j}=\min\{d_1,\lambda_j(R)\}$} and $\lambda_j(R)\leq d_1$. 

    Now consider the $k$'th inequality of the majorization relation in Eq.\,\eqref{eq:main}. By Ky Fan's maximum principle, there is a rank-$k$ projection $E$ s.t.\ the l.h.s.\ becomes $\sum_{l=1}^m\tr{E(A_l\otimes B_l)}$, which, by the above, is bounded by $\tr{\Lambda^T C}$, where $\Lambda_{r,j}\coloneqq\sum_{l=1}^m \lambda_r(A_l)\lambda_j(B_l)$. So $\Lambda$ contains the coefficients of the vector on the r.h.s.\ of Eq.\,\eqref{eq:main} when expanded in a product basis, and
    $$\sum_{l=1}^m\tr{E(A_l\otimes B_l)}\leq \tr{\Lambda^T C}\leq\max_{\substack{\Gamma_{rj}\in[0,1]\\ \sum_{r,j}\Gamma_{rj}=k}} \tr{\Lambda^T \Gamma} \leq \sum_{q=1}^k ({\rm vec\;}\Lambda)_q^\downarrow,$$
    which proves Eq.\,\eqref{eq:main} for positive matrices. Also note that in this case the sums of all entries of the two sides of Eq.\,\eqref{eq:main} agree, since both are equal to $\sum_l\tr{A_l}\tr{B_l}$ -- so majorization holds.

    It remains to show that the case of singular values of general matrices reduces to the case of positive matrices. We follow the argument of \cite[Lem.\,1]{RicoWolf}. So let $A_l$, $B_l$ now be arbitrary, set $M_l\coloneqq A_l\otimes B_l$ with polar decomposition $M_l=U_l|M_l|$, with $|M_l|=|A_l|\otimes |B_l|$, and define $$L\coloneqq\sum_{l=1}^m |M_l|^{1/2} U_l^*\otimes| l\rangle ,\qquad R\coloneq\sum_{l=1}^m |M_l|^{1/2}\otimes| l\rangle,$$
    where $\{|l\rangle\}_{l=1}^{m}$ denotes an orthonormal basis of $\C^m$. Then 
    $$L^*R=\sum_{l=1}^m M_l,\quad L^*L=\sum_{l=1}^m|M_l^*|,\quad R^*R=\sum_{l=1}^m|M_l|.$$
    The Cauchy-Schwarz inequality for unitarily invariant norms \cite[IX.5]{Bhatia} applied to the Ky Fan $k$-norm (sum of the top $k$ singular values) then yields
    \begin{equation}
\label{eq:svdfinalproof}\big\|L^*R\big\|_{(k)}\leq\big\|L^*L\big\|_{(k)}^{1/2}\big\|R^*R\big\|_{(k)}^{1/2}=\Big\|\sum_{l=1}^m |M_l^*|\Big\|_{(k)}^{1/2} \Big\|\sum_{l=1}^m |M_l|\Big\|_{(k)}^{1/2}.
    \end{equation}
Both norms on the r.h.s.\ of Eq.\,(\ref{eq:svdfinalproof}) are bounded via Eq.\,(\ref{eq:main}) applied to positive matrices. Moreover, both bounds lead to the same right-hand side, since $\lambda(|A_l|)=\lambda(|A_l^*|)=\sigma(A_l)$, and similarly for $B_l$, which completes the proof.
\end{proof}
\emph{Remark:} The extension of Eq.\,(\ref{eq:main}) to more tensor factors does not hold for $m>2$ \cite{Alhejji2026} but can be shown for $m=2$ \cite{Alhejji2026, Guterman}.
\begin{corollary*}[Singular values of quantum channels] Let $(K_l)_{l=1}^m\subseteq \C^{d\times d}$ be Kraus operators of a completely positive map $\C^{d\times d}\ni X\mapsto T(X)\coloneqq\sum_{l=1}^m K_l X K_l^*$. Then the singular values of $T$ satisfy the weak majorization
\begin{equation}
   \sigma(T)\prec_w \sum_{l=1}^m\sigma(K_l)\otimes\sigma(K_l)\Big. 
\end{equation}
    
\end{corollary*}
\begin{proof}
    When equipping $\C^{d\times d}$ with the Hilbert Schmidt inner product, the singular values of $T$ become those of $\sum_l  K_l\otimes \overbar{K_l}$ w.r.t.\ the ordinary inner product (cf. \cite[Eq.\,(4)]{Wolf2008}). Hence, the Theorem and $\sigma(\overbar{K_l})=\sigma(K_l)$ yield the result. 
\end{proof}

%
%

\appendix

\bibliographystyle{halpha}
\bibliography{KyFan}{}\vspace*{15pt}

\end{document}